\documentclass[11pt]{article}

\usepackage[letterpaper,margin=2.3cm]{geometry}
\usepackage{cite}
\usepackage{framed}
\usepackage[framemethod=tikz]{mdframed}
\usepackage{appendix}
\usepackage{graphicx}
\usepackage{xcolor}
\usepackage{varwidth}
\usepackage{wrapfig}
\usepackage{algorithm}
\usepackage[noend]{algpseudocode}
\usepackage[textsize=tiny]{todonotes}
\usepackage{array}

\usepackage{paralist}
\usepackage{amsthm}
\usepackage{amsmath,amssymb}

\usepackage{enumitem}
\setitemize{noitemsep,topsep=3pt,parsep=3pt,partopsep=3pt}

\usepackage{xspace}
\usepackage{xifthen}

\definecolor{darkgreen}{rgb}{0,0.5,0}
\definecolor{darkblue}{rgb}{0,0,0.8}
\usepackage{hyperref}
\hypersetup{
    unicode=false,          
    colorlinks=true,        
    linkcolor=darkblue,          
    citecolor=darkgreen,        
    filecolor=magenta,      
    urlcolor=cyan           
}
\RequirePackage[]{silence}
\usepackage[nameinlink,capitalize]{cleveref}

\newcommand{\calE}{\ensuremath{\mathcal{E}}}

\newcommand{\TC}{\ensuremath{\mathsf{TC}}}

\newcommand{\ignore}[1]{}

\newtheorem{theorem}{Theorem}[section]
\newtheorem{lemma}[theorem]{Lemma}

\newtheorem{definition}{Definition}[section]

\algnewcommand\algorithmicswitch{\textbf{switch}}
\algnewcommand\algorithmiccase{\textbf{case}}

\algdef{SE}[SWITCH]{Switch}{EndSwitch}[1]{\algorithmicswitch\ #1\ \algorithmicdo}{\algorithmicend\ \algorithmicswitch}%
\algdef{SE}[CASE]{Case}{EndCase}[1]{\algorithmiccase\ #1}{\algorithmicend\ \algorithmiccase}%
\algtext*{EndSwitch}%
\algtext*{EndCase}%

\newcommand{\ID}{\ensuremath{\mathrm{ID}}\xspace}

\newcommand{\set}[1]{\left\{#1\right\}}

\newcommand{\hide}[1]{}

\newcommand{\FullOrShort}{full}

\ifthenelse{\equal{\FullOrShort}{full}}{
	
  \newcommand{\onlyLong}[1]{#1}
  \newcommand{\onlyShort}[1]{}

}{
  
  \newcommand{\onlyShort}[1]{#1}
  \newcommand{\onlyLong}[1]{}

}

\title{The Communication Cost of Information Spreading\\ in\\ Dynamic Networks}
\author{
Mohamad Ahmadi\\
	\small{University of Freiburg}\\
	\small{mahmadi@cs.uni-freiburg.de}
\and
Fabian Kuhn\\
	\small{University of Freiburg}\\
	\small{kuhn@cs.uni-freiburg.de}
\and
Shay Kutten\\
	\small{Technion}\\
	\small{kutten@ie.technion.ac.il}
\and
Anisur Rahaman Molla\\
	\small{NISER}\\
	\small{molla@niser.ac.in}
\and
Gopal Pandurangan\\
	\small{University of Houston}\\
	\small{gopalpandurangan@gmail.com}
}

\date{}

\begin{document}

\maketitle              

\thispagestyle{empty}

\begin{abstract}
   This paper investigates the message complexity of distributed {\em
     information spreading} (a.k.a {\em gossip or token
     dissemination}) in adversarial dynamic networks. While
   distributed computations in dynamic networks have been studied
   intensively over the last years, almost all of the existing work
   solely focuses on the time complexity of distributed algorithms.

   In information spreading, the goal is to spread $k$ tokens of
   information to every node on an $n$-node network. We consider the
   {\em amortized} (average) message complexity of spreading a token,
   assuming that the number of tokens is large. In a static network,
   this basic problem can be solved using (asymptotically optimal)
   $O(n)$ amortized messages per token. Our focus is on
   token-forwarding algorithms, which do not manipulate tokens in any
   way other than storing, copying, and forwarding them.

   We consider two types of adversaries that have been studied extensively in dynamic networks: {\em adaptive} and {\em oblivious}. The adaptive {\em worst-case} adversary provides a dynamic sequence of network graphs under the assumption that it is aware of the
   status of all nodes and the algorithm (including the current random
   choices) and can rewire the network arbitrarily in every round with
   the constraint that it always keeps the $n$-node network
   connected.  On the other hand, the {\em oblivious} adversary  is a worst-case adversary
   that is oblivious to the random choices made
   by the algorithm.  The message complexity of information spreading is not yet fully understood in these models.  In particular, the central question that motivates our work is whether one can achieve subquadratic
   amortized message complexity for information spreading.
   
   We present two sets of results depending on how nodes
   send messages to their neighbors:

   \noindent 1. {\em Local broadcast:} We show a tight
   lower bound of $\Omega(n^2)$ on the number of amortized local
   broadcasts, which is matched by the naive flooding algorithm.

    \noindent 2. {\em Unicast:}  We study the message complexity 
    as a function of the number of dynamic changes in the
   network. To facilitate this, we introduce a natural  complexity measure for analyzing dynamic networks called
   {\em adversary-competitive message complexity}  where the adversary pays a
   unit cost for every topological change.  Under this model, it is
   shown that if $k$ is sufficiently large, we can obtain an optimal
   amortized message complexity of $O(n)$.  
   We also
   present a randomized algorithm that achieves {\em subquadratic}
   amortized message complexity when the number of tokens is
   not large
   under an {\em oblivious adversary}. Our analysis of the unicast communication under
   the adversary-competitive model (which may be of independent interest) is a main contribution
   of this paper. 
   
   Our work is a step towards fully understanding the message complexity of information spreading in dynamic networks.
\end{abstract}
\vspace{.5cm}

 \onlyShort{\vspace{-0.1in}}
\onlyLong{\section{Introduction}}
\onlyShort{\section{Introduction and Related Work}}
\label{sec:introduction}
     \onlyShort{\vspace{-0.1in}}

Many modern distributed communication networks such as ad hoc wireless, sensor, and mobile networks, overlay and peer-to-peer (P2P) networks are inherently dynamic (suffer from a  high rate of connections and disconnections) and
bandwidth-constrained.  Hence, understanding the possibilities and limitations of distributed computation in dynamic networks 
has been a major goal in recent years. 

In this paper, we study the fundamental problem of information
spreading on (synchronous) dynamic networks.  This problem was
analyzed for static networks by Topkis~\cite{topkis85}, and was in
particular studied on dynamic networks by Kuhn, Lynch, and
Oshman~\cite{kuhn-stoc10}.  In the information spreading problem (also
called {\em $k$-gossip or $k$-token dissemination}), there are $k$
pieces of information (tokens) that are initially present in some
nodes and the problem is to disseminate the $k$ tokens to all the $n$
nodes in the network, under the bandwidth constraint that one token
can go through an edge per round.  This problem is a fundamental
primitive for distributed computing; indeed, solving $n$-gossip, where
each node starts with exactly one token, allows any function of the
initial states of the nodes to be computed, assuming the nodes know
$n$~\cite{kuhn-stoc10}.

The dynamic network models that we consider in this paper allow a worst-case
adversary  known as {\em strongly adaptive}   that can choose any communication links among the nodes for each
round, with the only constraint being that the resulting communication
graph be connected in each round; this adversary can choose the links with the knowledge of the tokens that any node
can send in that round as well as its random choices (in one of the results we also consider an oblivious adversarial model). Our adversarial models are  closely related to those adopted in recent
studies (e.g.,  see
~\cite{avin08,santoro,sarmaMP15,Dutta13,kuhn-stoc10,odell05}). We
distinguish two variants of the basic model, depending on whether
nodes communicate by \emph{local broadcast} (i.e., a node always sends the same
message to all its neighbors) or whether we allow nodes to do \emph{unicast
communication} (i.e., nodes can possibly send different messages to
different neighbors in the same round). For more information on the model, we refer
to Section \ref{sec:model}. We note that most of the prior work (e.g.,
\cite{Dutta13,kuhn-stoc10,odell05}) only considered communication by
local broadcast.

The focus of the present paper is on {\em
  token-forwarding} algorithms, which do not manipulate tokens in any
way other than storing, copying, and forwarding them.
Token-forwarding algorithms are simple and easy to implement and have
been widely studied (e.g, see~\cite{leighton:book,peleg}).  The paper
investigates the {\em message complexity} of token-forwarding
algorithms for information spreading. Message complexity---the total
number of messages sent by all nodes during the course of an
algorithm---is an important performance measure. It directly relates
to the cost of communication, which is a dominant cost in many
real-world settings (e.g., it is correlated to energy, power, etc.\ in
wireless networks).  While information spreading in dynamic networks
have been studied intensively over the last years, almost all of the
existing work (e.g.,
\cite{gopal-disc16,avin08,sarmaMP15,Dutta13,HK12,kuhn-stoc10}) solely
focuses on the {\em time (round) complexity} of distributed
algorithms. (However, some works that focus on time complexity imply  bounds on messages --- see
e.g., \cite{BCF11,CCDFIPPS13,CS15}.) In many cases, the currently best algorithms for
information spreading in adversarial dynamic networks have a high
message complexity and in many cases, a high time complexity as well.
In contrast, in this paper,  we are interested in the {\em amortized} message complexity of information spreading, i.e.,
the average message cost of spreading $k$ tokens (when $k$ is large) in a dynamic network. To the best of our knowledge, this  aspect has not been studied in prior works on information spreading in dynamic networks (\onlyLong{cf. Section \ref{sec:related}}\onlyShort{cf. Section ``Related Work and Comparison'' in the full paper \cite{AKKMP18}}).

In any $n$-node static network, a simple token-forwarding algorithm
that pipelines token transmissions up a rooted spanning tree, and then
broadcasts them down the tree completes $k$-gossip in $O(n + k)$
rounds~\cite{peleg}, which is clearly asymptotically tight because the
diameter of the network might be $\Theta(n)$ and because every node has to
receive $k$ different tokens. In fact, $O(n+k)$ rounds are even
sufficient if in each round, each node forwards an arbitrary not yet
forwarded token to each of its neighbors \cite{topkis85}.  In a
dynamic network, it is known that under a strongly adaptive adversary
and if the communication is via local broadcast, the $O(n+k)$ bound
cannot be achieved; Dutta et al.~\cite{Dutta13} (see also \cite{HK12})
showed that $\Omega(nk/\log (nk) +n)$ rounds are necessary.  This
bound is essentially tight (up to a logarithmic factor), since one can
easily achieve an upper bound of $O(nk)$ by flooding.  We do not know
any tight bounds on the time complexity for {\em unicast} communication.

With regard to messages, we are interested in the {\em amortized}
(average) message complexity of spreading a token. In a static
network, one can first build a spanning tree (which can take as much
as $\Omega(n^2)$ messages\footnote{This bound is true in the KT0 model where nodes do not have initial knowledge
of their neighbors' IDs. On the other hand, in the KT1 model, where each node has initial knowledge of the IDs of their respective
neighbors, it is possible to build a spanning tree in $O(n \text{ polylog}(n))$ messages \cite{king}. Note that this distinction is not very important
in the amortized setting in a static network, since in both cases the amortized message complexity is $O(n)$ if $k = \Omega(n)$.
In the dynamic setting, we essentially assume the KT1 model for unicast communication, whereas for broadcast communication, the distinction is not important, see Section \ref{sec:model} for more details.}  in graphs with $\Theta(n^2)$
edges \cite{kutten-jacm15}), and then using the spanning tree edges to
disseminate the tokens to all nodes; this takes $O(n^2+nk)$ messages
overall or $O(n^2/k + n)$ amortized messages per token. If $k$ is
sufficiently large\footnote{There are natural applications where $k$ is large,
  e.g., if all nodes have tokens to broadcast or if some node has
  a stream of messages as, for example, in audio/video transmissions.}, say at least $n$, then
the above bound gives $O(n)$ amortized messages per token, which is
optimal (since each node has to receive the token). On the other hand,
for dynamic networks, the situation is far less clear. In the case of
{\em local broadcast} communication (where each broadcast is counted
as one message\footnote{This is reasonable, especially, in the context
  of wireless networks where nodes communicate by local broadcast.}), an
$O(n^2)$ amortized message upper bound per token is
straightforward to obtain by using flooding (each
node broadcasts each token for $n$ rounds). For unicast communication
(cf.~Section~\ref{sec:model}), again an $O(n^2)$ amortized upper bound
is easy to obtain (each node sends each token at most once to each
other node; note that for unicast communication each message to a neighbor is counted as one message). In both cases, 
non-trivial lower bounds are not known.  Thus, the {\em central question
  that we seek to address in this work is whether one can achieve
  $o(n^2)$ or even asymptotically optimal
  $O(n)$ amortized message complexity when $k$ is large (for both
  the local broadcast and the unicast settings)}. We note that prior works (including \cite{gopal-disc16,avin08,sarmaMP15,Dutta13,HK12,kuhn-stoc10}) do not address this question.
  
 \onlyShort{\vspace{-0.2in}}
\subsection{Our Main Results} 
     \onlyShort{\vspace{-0.1in}}
In the local broadcast setting, we give a negative answer to the above
question and show that with a strongly adaptive adversary, the
$\Theta(n^2)$ amortized message complexity bound of the naive
algorithm is indeed necessary
(cf.~Section~\ref{sec:local-broadcast}). This ``bad'' bound for local broadcast is a  motivation for considering the (more challenging) unicast setting. For the unicast setting, we
study how the message complexity behaves as a {\em function of the
  number of dynamic changes in the network}. To facilitate this, we introduce a new and natural  complexity measure for analyzing dynamic networks called
   {\em adversary-competitive message complexity} (cf. Definition \ref{def:advcomp_msgcompl}).  While the adversary is
free to change the topology arbitrarily from round to round, this measure allows one to 
intuitively assume that it has to pay some price for every connection
and reconnection and we allow an algorithm a ``free'' communication
budget of comparable size. This measure has natural real-world motivation.  For example, in real-world communication networks, due to the actions of
  the lower layer link protocol (that is responsible to establish the
  connection when a physical link comes up), one can assume that
  whenever a new edge is created, some information is exchanged anyhow
  by the link layer. Thus, it is reasonable to assume that there is some cost to be paid in establishing or re-establishing a link (say, after the link is down for a while). Our new measure formalizes this intuition.
  
  Under the new complexity measure (defined formally in Section
\ref{sec:model}), we show that if $k$ is sufficiently large, we obtain
an optimal amortized message complexity of $O(n)$
(cf.~Section~\ref{sec:adversary-pays}).  In case the dynamic network
topology satisfies some natural additional properties, we also show
that the algorithm terminates in $O(nk)$ rounds. We present two
algorithms in this setting depending on how the tokens are initially
distributed: (1) a {\em single-source case}, where all the tokens
start at the same node and (2) a {\em multi-source case}, where the
initial token distribution is arbitrary.
  
When the number of tokens is not very large, say $k =n$ (i.e.,
$n$-gossip), the $O(n)$ amortized bound does not hold.  In this
setting, we are able to show a {\em subquadratic} amortized message
complexity under an {\em oblivious adversary}, which is same as the
worst-case adversary, except that it is oblivious to the random
choices made by the algorithm and the execution history (cf. Section
\ref{sec:oblivious-adv}). Our algorithm is randomized and is based on
random walks.

Our analysis of the unicast communication under the adversary-competitive model  is a main contribution
   of this paper. We believe that the adversary-competitive model can be an useful alternative to the current models in analyzing various other important problems such as leader election and agreement in dynamic networks (see e.g., \cite{podc13,disc15}).

 Our work raises several key open questions that are discussed in Section \ref{sec:conclusion}.
 \onlyShort{For lack of space, additional related work, full proofs  and additional details are deferred to the full paper \cite{AKKMP18}}
 
 \onlyLong{
\subsection{Related Work and Comparison}
\label{sec:related}
Information spreading (or dissemination) in networks is a fundamental
problem in distributed computing with a rich literature. The
problem is generally well-understood on static networks, both for
interconnection networks~\cite{leighton:book} as well as general
networks~\cite{attiya+w:distributed,lynch:distributed,peleg}.  In
particular, the $k$-gossip problem can be solved in $O(n + k)$ rounds
on any $n$-node static network~\cite{topkis85}.  There are also several papers on broadcasting, multicasting, and related
problems in static heterogeneous and wireless networks (e.g.,
see~\cite{alon+blp:radio,bar-noy+gns:multicast,bar-yehuda+gi:radio,clementi+ms:radio}).

Dynamic networks have been studied extensively over the past three
decades.  Early studies focused on dynamics that arise when edges or
nodes fail (but, generally don't consider edges/nodes recovering from failures).  A number of fault models, varying according to extent and
nature (e.g., probabilistic vs.\ worst-case) of faults and
the resulting dynamic networks have been analyzed (e.g.,
see~\cite{attiya+w:distributed,lynch:distributed}).  There are
several studies that constrain the rate at which changes occur or
assume that the network eventually stabilizes (e.g.,
see~\cite{afek+ag:dynamic,dolev:stabilize,gafni+b:link-reversal}).

To address highly unpredictable network dynamics,
models with stronger adversaries have been studied
by~\cite{avin08,kuhn-stoc10,odell05} and others;
see the recent survey of \cite{santoro} and the references therein.
Unlike prior models on dynamic networks, these models and ours do not
assume that the network eventually stops changing; the algorithms are
required to work correctly and terminate even in networks {\em that change
continually over time}.  
 
 The model of~\cite{Dutta13,HK12,kuhn-stoc10} allows for a much stronger
   adversary than the ones considered in past
   work~\cite{awerbuch+bbs:route,awerbuch+bs:anycast,awerbuch+l:flow}.
   In particular, the work of \cite{Dutta13} (also see \cite{HK12}), showed  that every token forwarding information spreading algorithm that uses local broadcast for communication under a strongly adaptive adversary (the same as considered in this paper --- cf. Section \ref{sec:local-broadcast}) requires
$\Omega(n^2/\log n)$ rounds to complete.
    The
 survey of~\cite{kuhn-survey} summarizes recent work on dynamic
 networks (see also the early works of \cite{CPMS07,CMMPS08}).

Recent work of~\cite{haeupler:gossip,haeupler+k:dynamic} presents
information spreading algorithms based on network
coding~\cite{ahlswede+cly:coding}.  As mentioned earlier, one of their
important results is that the $k$-gossip problem on the adversarial
model of~\cite{kuhn-stoc10} can be solved using network coding in
$O(n+k)$ rounds assuming the token sizes are sufficiently large
($\Omega(n\log n)$ bits). 

 It is important to note that {\em all the above results} deal with
 the {\em time complexity} of information spreading in dynamic
 networks (i.e., the number or rounds needed)
and not with the message complexity.  The focus here is on {\em amortized} message complexity for spreading
$k$ tokens. We note that there is an important difference between the two measures. In particular, algorithms with efficient time complexity need not necessarily be message-efficient and vice-versa and hence prior time complexity-based results do not directly
imply the results of this paper. 
Indeed, one can  exchange up to $\Theta(n^2)$ messages (in a graph with $\Theta(n^2)$ edges)
in just one round,  and since one needs at least $\Omega(n)$ rounds for information spreading (in the worst-case), the total message 
complexity can be as high as $\Omega(n^3)$ (for unicast). In other words, a message-efficient algorithm can take a longer time but exchanging
less total number of messages, e.g., by sending
messages only along a few edges and/or by using silence. However, as we show in Section \ref{sec:local-broadcast}, the amortized message complexity lower bound (even)  for local broadcast (where a node's local broadcast to all its neighbors is counted as just one message)  is close to the worst possible, i.e., $\Omega(n^2/\log n)$. The proof for this lower bound is inspired by the time complexity lower bound of \cite{Dutta13}, although the two proofs differ in their details. The ``bad'' lower bound for local broadcast motivates considering unicast communication which is the main focus of this paper. \onlyLong{It is important to point out  another difference
between amortized time complexity and amortized message complexity. While amortized time complexity can be as low as
$\Omega(D)$ (where $D$ is the network diameter, which can be much smaller than $n$), the amortized message complexity
is at least (trivially) $\Omega(n)$, since a token has to reach all the $n$ nodes.}  There has not been much progress  on improving
time complexity (total or amortized) in dynamic networks (both for unicast and local broadcast) in the oblivious adversary model in general networks, although
prior works \cite{gopal-disc16, sarmaMP15} has achieved improved (subquadratic in $n$) total {\em time} complexity under additional 
assumptions on the dynamic network model (these are  different from what is considered here). In particular, the work 
of \cite{sarmaMP15} considers a  dynamic network and presents an information spreading algorithm that can have
subquadratic  {\em time} complexity under some restricted conditions, e.g., when the dynamic mixing time (defined in \cite{sarmaMP15}) is small. The work of \cite{sarmaMP15} does not address amortized message complexity at all and the result
in the oblivious adversary setting of this paper  do not follow from the results of \cite{sarmaMP15}.  Both papers
use techniques based on random walks (which are very useful in the oblivious setting) which were originally developed in \cite{podc10,jacm}, but the  algorithms are quite different.

While the work of \cite{Dutta13} adopts the strongly adversarial model and local broadcast communication (we adopt the same model here for the local broadcast communication ---cf. Section \ref{sec:local-broadcast}), the work of \cite{ gopal-disc16,sarmaMP15}
adopts the oblivious adversary model (we also adopt the oblivious model here for unicast communication in Section \ref{sec:oblivious-adv}),
a {\em novel aspect of this paper is introducing and adopting a new communication cost model} that measures the communication cost of an algorithm as a function of the amount of topological changes that occur in a given execution and a new message complexity measure called {\em adversary-competitive message complexity} (Section \ref{sec:model}). A main contribution of this paper
shows that under this new complexity measure, one can obtain an efficient amortized message complexity for {\em unicast} communication that is significantly better  than the worst-case bound of $\Omega(n^2)$.
Our new measure is inspired by and  related to the notion of {\em resource competitive} algorithms \cite{saia}, although the 
details are different. The previous measure does not address an adversary in the context of dynamic networks.
}


\onlyShort{\vspace{-0.1in}}
\subsection{Dynamic Network, Communication, and Cost Model}
\label{sec:model}
\onlyShort{\vspace{-0.1in}}
In the following, we formally define the dynamic network model, the
communication models we consider, as well as the way in which we
measure the communication cost (or message complexity) of a given
token dissemination algorithm.
\\[.1cm]
\noindent{\bf Dynamic Network Model:} We model the network as a synchronous
dynamic graph $G$ with a fixed set of nodes $V$. Nodes communicate in
synchronous rounds where round $r$ starts at time $r-1$ and ends at
time $r$. For any integer $r\geq 1$, we use $G_r=(V, E_r)$ to denote
the graph of round $r$.  Throughout, we use $n:=|V|$ to denote the
number of nodes and $m_r:=|E_r|$ to denote the number of edges in
round $r$. For convenience, we define $E_0:=\emptyset$ and thus $G_0$
is the empty graph $(V,\emptyset)$. For every $r\geq 0$, we call
$E_r^+:=E_r\setminus E_{r-1}$ the set of edges \emph{inserted in round
  $r$} and we call $E_r^-:=E_{r-1}\setminus E_r$ the set of edges
\emph{removed in round $r$}.

In order to always allow progress when globally broadcasting a
message, we assume that each graph $G_r$ is connected for $r \geq 1$.

We sometimes also need the property that every edge which gets
inserted remains in the graph for at least a given number of rounds.
For an integer $\sigma\geq 1$, we call a graph \emph{$\sigma$-edge
  stable} if for every $r\geq 1$ and every edge $e\in E_r$, there
exists a round $r'\geq \max\set{1, r-\sigma+1}$ such that
$e\in E_{r'}\cap\dots\cap E_{r'+\sigma-1}$. Hence, after it appears,
every edge remains in the graph for at least $\sigma$ consecutive
rounds. Note that every dynamic graph is $1$-edge stable.

We assume that the dynamic topology is provided by a worst-case
adversary. There are adversaries of different strengths, depending on
the capability of adaptively reacting to random choices of a given
algorithm. In this paper, we distinguish between a \emph{strongly
  adaptive adversary} and an \emph{oblivous adversary}. The strongly
adaptive adversary knows the algorithm's randomness of the current
round in order to determine the dynamic topology for that
round\footnote{In comparison, a \emph{weakly adaptive adversary} only
  knows the algorithm's randomness up to the round before the current round.}. The
oblivious adversary is oblivous to any randomness used by the
algorithm and to any decision made by the algorithm, i.e., it has to commit to the sequence of network
topologies before the execution of a distributed algorithm
starts. Note that for deterministic algorithms, both adversaries are
the same.
\\[.1cm]
\noindent {\bf Communication Model:} Throughout the paper, we assume that each
node $v\in V$ has a unique $O(\log n)$-bit identifier $\ID(v)$ and
that in each round, each node can send messages containing a {\em constant
number of tokens} and $O(\log n)$ additional bits to its neighbors.

We distinguish different modes of communication, depending on whether the message exchange among
neighbors is based on local broadcast or on unicast. 
\\[.1cm]
\noindent {\bf 1. Local Broadcast Communication:}
  In each round $r$, each node $v$ can locally broadcast a message
  which is received by all neighbors of $v$. Node $v$ learns the set
  of neighbors in round $r$ when receiving the round $r$ messages from
  them. 
  \\[.1cm]
  \noindent {\bf 2. Unicast Communication:}
  
  At the beginning of each round $r$,
    each node $v$ is informed about the IDs of its neigbors in round
    $r$. Node $v$ can then send a different message to each neighbor.  
Note that if the neighborhood information is not available instantaneously, it can be obtained by exchanging messages. As a consequence, in a $2$-edge stable dynamic graph, the known neighborhood information and unknown neighborhood information are equivalent with a cost of extra messages.
\\[.1cm]
\noindent {\bf Communication Cost:} The \emph{communication cost} of a protocol
is measured by its \emph{message complexity}, i.e., by the total
number of messages sent by all the nodes throughout the whole
execution. 

\begin{definition}[Message Complexity]
  The \emph{message complexity} of a distributed algorithm is the \emph{total
  number of messages} sent in a worst-case execution. If communication
  is by local broadcast, each local broadcast by some node counts as
  one message. If communication is by unicast, messages to different
  neighbors are counted separately.  
\end{definition}
The main focus of this article
is to study the message complexity of solving the token dissemination
problem.
\begin{definition}[$k$-Token Dissemination Problem]\label{def:problem-definition}
 For some positive integer $k$, $k$ distinct tokens are initially placed at some nodes in the network. The goal is to disseminate all the $k$ tokens to all the nodes in the network.  
\end{definition}

As discussed in Section \ref{sec:introduction}, we are particularly
interested in understanding to what extent dynamic topology changes affect the
communication cost of token dissemination. We  thus consider a
cost model that measures the communication cost of an algorithm as a
function of the number of topological changes.  We formally define the \emph{number of topological
  changes} $\TC(\calE)$ of an execution $\calE$ as the total number of
edges that are inserted throughout an execution, i.e., for an $x$-round
execution $\calE$ with dynamic graph $G_r=(V,E_r)$, we have
$\TC(\calE):=\sum_{r=1}^x |E_r^+|$.\footnote{Note that since we assume
  that at time $0$ we start with an empty graph, the total number of
  edge deletions is always upper bounded by the total number of edge
  insertions. Hence we only count the edge insertions and not the edge deletions.} The following definition captures the notion that for
each dynamic change caused by the dynamic network adversary, a
distributed algorithm is allowed to send a given number of messages
``for free''.
\begin{definition}[Adversary-Competitive Message Complexity]\label{def:advcomp_msgcompl}
  Given a parameter $\alpha\geq 0$, we say that a distributed
  algorithm has \emph{$\alpha$-adversary-competitive message
    complexity $M$} if for every execution $\calE$, the total message
  complexity of the algorithm is upper bounded by
  $M + \alpha\cdot \TC(\calE)$.
\end{definition}

To capture the progress of an algorithm, one way is to count how many
new tokens have been received so far by the nodes. 

\begin{definition}[Token Learning]
	A \textit{token learning} is an event $\langle v,\tau,r
        \rangle$ that occurs in some $x$-round execution $\calE$ if and only if node $v$ receives token $\tau$ for the first time in round $r$, where $r\leq x$. 
	Then, we say $v$ learns $\tau$ in round $r$. 
\end{definition}

Based on the above definition, if each of the $k$ tokens is initially given to exactly one of the $n$ nodes, it is trivial that $k(n-1)$ token learnings must occur during an algorithm execution solving $k$-token dissemination.


\onlyShort{\vspace{-0.3cm}} 
\section{Local Broadcast Model}
 \onlyShort{\vspace{-0.1in}}
\label{sec:local-broadcast}
Before we go to the unicast setting, which is the main focus of this paper,
we present a tight quadratic (in $n$) lower bound for the amortized message
complexity of disseminating $k$ tokens in the local broadcast
setting. 

We assume that each of the $k$ tokens can
initially be given to an arbitrary subset of the nodes with the only
restriction that the nodes initially have at most $k/2$ tokens on
average. We further assume that $k$ is at most polynomially large in
$n$. Our lower bound is an extension of the time complexity lower
bound, which was developed by Dutta et al.\ in \cite{Dutta13} and
which was slightly generalized and simplified in \cite{HK12}. The main
idea of the lower bound is as follows. If initially, each token is
given to each node independently with a constant probability, the
lower bound shows that in each round of any $k$-token dissemination
algorithm execution, a strongly adaptive adversary can enforce that in
total at most $O(\log n)$  tokens are learned by the nodes. Because
by the end of an execution, the nodes together need to learn
$\Theta(nk)$ tokens (each node needs to learn the tokens it does not
know initially), this directly implies a $\Omega(nk/\log n)$ time
complexity lower bound. Here, we adapt the technique of the lower
bound of \cite{Dutta13,HK12} to show that in any round with at most
$O(n/\log n)$ broadcasting nodes\footnote{Throughout this section, we
  call a node that performs a local broadcast in some round $r$, a
  broadcasting node in round $r$.}, a strongly adaptive adversary can
prevent any new tokens from being learned. Because the nodes together
need to learn $\Theta(nk)$ tokens, together with the upper bound of
$O(\log n)$ on the number of tokens learned in a single round, this
implies that a strongly adaptive adversary can force any token
dissemination algorithm to require at least $\Omega(nk/\log n)$ rounds
with at least $\Omega(n/\log n)$ broadcasting nodes. 
	This leads to the
overall message complexity of at least $\Omega(n^2k/\log^2n)$. 

	  To prove our lower bound, we mostly use the notation
 in \cite{HK12}. Let
$\mathcal{T}$ denote the set of $k$ tokens, and for each node
$v\in V$, let $K_v(t)$ be the set of tokens that node $v$ knows by
time $t$. In each round $r$, let $i_v(r)$ denote the token
broadcast by node $v$ if $v$ is a broadcasting node in round $r$. If
$v$ is not a broadcasting node in round $r$, we define
$i_v(r):=\bot$. Note that a strongly adaptive adversary can determine
the dynamic graph topology of round $r$ after each node has chosen the token
$i_v(r)$ to locally broadcast. Generally, a collection of pairs
$\big( v,i_v\big)$, where $v\in V$ and
$i_v\in \mathcal{T}\cup\set{\bot}$ is called a token assignment.

In addition, the adversary determines a token set
$K'_v\subseteq \mathcal{T}$ for each node $v$. The sets $K_v'$ are
just used for the analysis. Informally, one can think of $K_v'$ as an
additional set of tokens that node $v$ knows at time $0$. Formally, we
do not assume that node $v$ knows the tokens in $K_v'$ initially, but
whenever $v$ learns a token from $K_v'$, we do not count this as
progress (i.e., for node $v$, we only count how many tokens from
$\mathcal{T}\setminus K_v'$ it has learned). To formally measure the
progress, we define a potential function
$\Phi(t) := \sum_{v\in V} |K_v(t)\cup K'_v|$. Recall that we assume
that initially on average, each node knows at most $k/2$ tokens, i.e.,
$\sum_{v\in V} |K_v(0)|\leq nk/2$. The adversary chooses the sets
$K_v'$ in such a way that $\Phi(0)\leq 0.8nk$. In order to solve the
token dissemination problem, the potential has to grow to $nk$. The
choice of the sets $K_v'$ therefore guarantees that the potential
needs to grow by at least $0.2nk$ throughout the execution of a
$k$-token dissemination protocol.
	
To study the growth of the potential function, the following notion is used.
An (potential) edge $\set{u,v}$ is called \emph{free} in round $r$, if and only if the
communication over $\set{u,v}$ does not contribute to $\Phi(r) - \Phi(r-1)$,
i.e., $\set{u,v}$ is free if and only if $i_u(r)\in \set{\bot}\cup K_v(r-1)\cup K_v'$
and $i_v(r)\in \set{\bot}\cup K_u(r-1)\cup K_u'$. Otherwise, the edge $\set{u,v}$ is
called \emph{non-free}. When determining the topology of round $r$, a
strongly adaptive adversary can always add all free edges to the graph
$G_r$ 
 without causing any increase of the potential function. If after
adding all free edges, the graph has $\ell$ connected components, the adversary
needs to add $\ell-1$ additional edges ``non-free'' edges in order to make $G_r$
connected. The potential function can then grow by at most $2(\ell-1)$
because over each of these additional $\ell-1$ edges, one token can be
learned in each direction. In \cite{Dutta13,HK12}, it is shown using a probabilistic method that the sets $K_v'$ can be chosen such that
$\Phi(0)\leq 0.8nk$ and such that 
in each round, the graph induced by only the free edges has at most
$O(\log n)$ connected components. Every algorithm therefore needs at
least $\Omega(nk/\log n)$ rounds for the potential to grow to $nk$.

The following lemma from \cite{HK12} shows that if each token is
randomly added to each set $K_v'$ independently with probability
$1/4$, adding all free edges reduces the number of components to
$O(\log n)$ for all rounds with constant probability.

\begin{lemma}\label{lem:time-lb}(Lemma 1 of \cite{HK12}) If each set
  $K'_u$ contains each token $i\in\mathcal{T}$ independently with
  probability $1/4$, with probability at least $3/4$, for all rounds $r$ and all possible token
  assignments $(v, i_v(r))$ in round $r$, the graph $F(r)$
  induced by all free edges in round $r$ has at most $O(\log n)$
  connected components. 
\end{lemma}

We next show that if the number of broadcasting nodes is
small, adding all free edges leaves only one connected component. For
a constant $c>0$, we define a token assignment $(v,i_v)$ to be
$c$-sparse if at most $n/(c\log n)$ of the nodes are broadcasting nodes
(i.e., for at most $n/(c\log n)$ nodes, we have $i_v\neq
\bot$).

\begin{lemma}\label{lem:msg-lb}
  There is a constant $c>0$ such that if each set $K'_u$ contains each
  token $i\in\mathcal{T}$ independently with probability $1/4$, with
  probability at least $1-2^{-n}$, for all rounds $r$ and all possible
  $c$-sparse token assignments $(v, i_v(r))$, the graph $F(r)$ 
  induced by all free edges in round $r$ consists of a single
  connected component.
\end{lemma}
\begin{proof}
  We first bound the probability for a fixed $c$-sparse token
  assignment $(v, i_v)$. The claim of the lemma will then follow by a
  union bound over all the possible $c$-sparse token assignments.  Let
  $B$ denote the set of broadcasting nodes, i.e., the nodes for which
  $i_v\neq \bot$. Further, let $\beta := |B|\leq n/(c\log n)$ and let
  $\bar{B}:=V\setminus B$.  Clearly, all the edges among the nodes in
  $\bar{B}$ are free. It is therefore sufficient to show that for each
  node $v$ in $B$, there is a free edge connecting $v$ to a node in
  $\bar{B}$.  Then, all the free edges induce a connected graph over
  all the nodes (also see Figure \ref{fig:lb}).

  Consider an edge $\{u,v\}$, where $u\in \bar{B}$, $v\in B$, and $v$
  is locally broadcasting token $\tau$.  Edge $\set{u,v}$ is a free
  round (for every round $r$) if $\tau\in K'_u(0)$. This happens with
  probability $1/4$ (independently for every node $u\in \bar{B}$).  The probability that $v$
  has no free edge to some node in $\bar{B}$ is thus at most
  $(1/4)^{n-\beta}$.  Thus, the probability that there exists at least
  one node in $B$ that has no free edge to $\bar{B}$ is at most
  $\beta/4^{n-\beta}$.  Considering a union bound over all
  ${{n}\choose{\beta}} <n^\beta$ ways to choose a set of $\beta$ nodes
  and all at most $k^\beta$ ways to choose the tokens to be sent out
  by these nodes, the probability that there exists a token assignment
  for which there is a node in $B$ that has no free edge to $\bar{B}$
  can therefore be upper bounded by
  \begin{align*}
    \Pr(\exists v\in B \textit{ s.t. } & \forall u\in \bar{B}: \{v,u\} \textit{ is non-free})\\
                                       & \leq n^\beta \cdot k^\beta \cdot \frac{\beta}{4^{n-\beta}}\\
                                       & = 2^{\beta (\log(nk)+2) + \log{\beta-2n}}\\
                                       & \leq 4^{\frac{c}{2}\beta \log n-n} \quad \quad \text{[for some constant $c$]}\\
                                       & < 2^{-n} \quad \quad  \quad \quad \quad \text{[for } \beta<\frac{n}{c \log n}\text{]}
  \end{align*}	
  \noindent
  Hence, with probability at least $1-2^{-n}$, for each possible token
  assignment (and for each round), each node $v\in B$ has a free edge
  connected to some node in $\bar{B}$.
\end{proof}
\begin{figure}
	\centering
	\footnotesize
\begin{tikzpicture}[scale=0.85]
	\draw[dashed] (-2.5,0) ellipse (1.5cm and 2cm);
	\draw (2.5,0) ellipse (1.5cm and 2cm);
	
	\draw[fill] (-2.5,.3) circle (.07cm);
	\draw[fill] (-2.5,.8) circle (.07cm);
	\draw[fill] (-2.5,1.3) circle (.07cm);
	\draw[fill] (-2.5,-1.2) circle (.07cm);
	\draw[fill] (-2.5,-.2) circle (.07cm);

	\draw[fill] (-2.5,-0.9) circle (.02cm);
	\draw[fill] (-2.5,-0.8) circle (.02cm);
	\draw[fill] (-2.5,-0.7) circle (.02cm);

	\begin{scope}[shift={(2.5,0)}]
		\foreach \x in {30, 90,165, 210,270, 330}
		{
			\draw[fill]  (\x:1cm) circle (.07);
		}
	 \end{scope}
	 	
	\begin{scope}[shift={(2.5,0)}]
		\foreach \x in {30, 90,165, 210,270, 330}
		{
			\foreach \y in {30, 90,165, 210,270, 330}
				\draw[gray]  (\x:1cm) -- (\y:1cm);
		}
	\end{scope}
	 
	 \draw[gray] (-2.5,.3) -- (2.5,1);
	 \draw[gray] (-2.5,1.3) -- (2.5,1);
	 \draw[gray] (-2.5,.8) -- (1.6,.3);
	 \draw[gray] (-2.5,-.2) -- (1.7,-.5);
	 \draw[gray] (-2.5,-1.2) -- (2.5,-1);
	
	\node at (0,1.7) {\textit{free edges}};

	\node at (-2.5,2.3) {$B$};
	\node at (2.5,2.3) {$\bar{B}$};
\end{tikzpicture}

\caption{\footnotesize It shows the connected graph induced by (a subset of) the free edges in a round with at most $O(n/\log n)$ broadcasting nodes. 
	The free edges among the nodes in $\bar{B}$ induce a clique, and each of the broadcasting nodes in $B$ is connected to some node in $\bar{B}$ by a free edge.}
\label{fig:lb}
\end{figure}
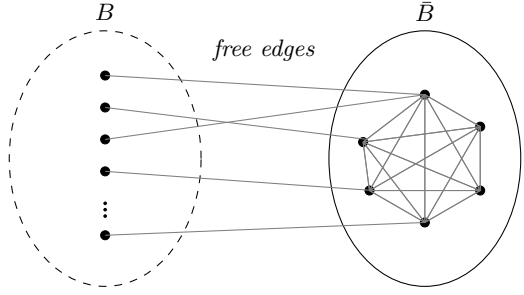
\onlyShort{\vspace{-0.3cm}} 
\begin{theorem}
  In any always connected dynamic network, if initially each node on
  average knows at most half of the $k$ tokens, the amortized message
  complexity of solving the $k$-token dissemination problem against a
  strongly adaptive adversary is at least $\Omega(n^2/\log^2n)$ in the
  local broadcast communication model.
\end{theorem}

\begin{proof}
  Using the probabilistic method, we show that the adversary can
  choose the sets $K_u'$ such that at time $0$,
  $\Phi(0)\leq 0.8nk$ and such that for every possible strategy of the
  algorithm, the adversary can choose the graph of each round such
  that (1) the graph is connected, (2) the number of connected
  components after adding all free edges is at most $O(\log n)$, and
  (3) if there are at most $n/(c\log n)$ broadcasting nodes, for a
  sufficiently large constant $c>0$, the free edges induce a connected
  graph. The theorem then follows because (a) the potential needs to
  grow by $0.2nk$ in order to solve the token dissemination problem
  and (b) the potential increase per round is always at most
  $O(\log n)$ and it is $0$ if the number of broadcasting nodes is less
  than $n/(c\log n)$.  \onlyShort{We defer the detailed proof to the full paper \cite{AKKMP18}.}

\onlyLong{
  To apply the probabilistic method, we let each set $K'_u$ contain
  each token $i\in\mathcal{T}$ independently with probability
  $1/4$. First note that by a standard Chernoff argument, the
  probability that $\sum|K'_u|> 0.3nk$ is exponentially small in $nk$
  and thus the probability that $\Phi(0)>0.8nk$ is also exponentially
  small in $nk$. Further, by
  Lemma~\ref{lem:time-lb} and Lemma~\ref{lem:msg-lb}, for every round $r$, and
  every token assignment $(v,i_v(r))$, the graph $F(r)$ induced by all
  the free edges has the following two properties with probability at
  least $\frac{3}{4}-2^{-n}$: (1) $F(r)$ contains at most
  $O(\log n)$ connected components, (2) $F(r)$ is connected over all
  the nodes if there are at most $n/(c\log n)$ broadcasting nodes.
  This
  shows that (for sufficiently large $n$), there is a way to choose
  the sets $K'_u$ sets such that $\Phi(0)\leq 0.8nk$, the potential
  increase per round is at most $O(\log n)$, and if there are at most
  $n/(c \log n)$ broadcasting nodes, the potential increase is $0$
  and the claim of the theorem follows.
  }
\end{proof}


\onlyShort{\vspace{-0.4in}} 
\section{Unicast Model}
\label{sec:adversary-pays}
\onlyLong{We want to solve  the $k$-token dissemination problem where the $k$ tokens are initially distributed (arbitrarily) over the network and the goal is to disseminate all the tokens to all the nodes with as few messages as possible.
To solve this problem, it turns out that it is first easier to consider a special instance --- called the {\em Single Source Case} --- where all the $k$ tokens are initially located in a single source node. We use the Single Source Algorithm  (Section \ref{sec:single-source}) as a subroutine to solve the more general 
{\em Multi-Source} case (Section \ref{multi-source}).}

 \onlyShort{\vspace{-0.1in}}
\subsection{Single Source Node}
\label{sec:single-source}
Consider the $k$-token dissemination problem such that all the $k$ tokens are initially given to a single source node.  Let us now present a deterministic algorithm to solve this problem with message complexity of $O(n^2+nk)+\TC(\calE)$ against a strongly adaptive adversary.  Hence, the algorithm has \emph{1-adversary-competitive  (total) message complexity} of $O(n^2+nk)$ (cf.~Def.~\ref{def:advcomp_msgcompl}).  In other words, if the algorithm is provided with a budget that equals to the number of topological changes, then for sufficiently large $k$, the amortized message complexity to disseminate the tokens is linear in $n$.  Note that even in a static graph, the cost to disseminate a single token is $\Omega(n)$.  Hence, if the number of tokens is at least linear in $n$, the amortized message complexity is asymptotically best possible.
	Before we present the algorithm and its analysis, consider the following definitions.
\begin{definition}[Complete and Incomplete Node]\label{def:complete-node}
	We say that node $v$ is \textit{complete} at time $t$ if it has all the $k$ tokens at this time. Otherwise, $v$ is \textit{incomplete}.
 \end{definition}
\begin{definition}[Bridge Node]
\label{def:bridge-node}
	In each round, any incomplete node that has a complete neighbor is called a \textit{bridge node} for that round. 
\end{definition} 

\onlyShort{\vspace{-0.2in}}
\subsubsection{Single-Source Unicast Algorithm}
\onlyShort{\vspace{-0.1in}}
\label{sec:adv-pays-single-source-algo}
	The source node considers an arbitrary order of the tokens and assigns integer $i$ to its $i^{th}$ token as its \textit{token ID}.
	In the algorithm, only complete nodes send tokens during an execution. 
	To this end, each complete node announces its completeness to its neighbors.
	In each round, each incomplete node sends token requests to (some of) its complete neighbors.
	Then, in the very next round, each complete node sends back the requested tokens to the requesting nodes if it is still connected to them. 
	Although the general idea is simple, a careful strategy is needed to avoid redundant communication.
	
	Each complete node $v$ informs each node about the time of $v$'s completeness at most once by remembering which nodes $v$ informed before.
	Each node also remembers all the complete nodes it is informed by about their completeness.  
	Each incomplete node chooses among its complete neighbors for sending token requests to, based on a priority defined by the following categorization of its adjacent edges.
	
	Consider an edge $e=\{v,w\} \in E_r$ such that $v$ is incomplete and $w$ is complete.
	Then $e$ is called \textit{new} in round $r$ if the edge is inserted at the beginning of round $r$ or $r-1$.
	Edge $e$ is called \textit{contributive} if it is not new, but a new token is sent over it between the last insertion of the edge and the end of round $r$, i.e., it contributes to the dissemination. 
	Otherwise, if $e$ is neither new nor contributive, it is called \textit{idle} in round $r$.
	
	Based on the above definitions, if $v$ has $\tau$ missing tokens, it creates $\tau$ token requests, one for each missing token. 
	Then, $v$ assigns exactly one distinct token request to each of the new edges (if any).
	Afterwards, if there are still token requests left to be assigned, $v$ assigns exactly one request to each of the idle edges (if any). 
	Finally, $v$ does the same for the contributive edges.
	Note that as each edge has at most one assigned token request, there might be token requests that are not assigned in the current round.
	At the end, $v$ sends the assigned token requests in round $r$ over the corresponding edges. 
		
	Note that for categorizing an adjacent edge $e=\{v,w\}$, an incomplete node $v$ might need to know whether it learns a token over $e$ in round $r$ or not.
	However, if $v$ sends a token request over $e$ in round $r-1$, and $e\in E_r$, then $v$ knows that it learns a token over $e$ in round $r$. 
	Moreover, to avoid sending redundant token requests, node $v$  needs to know whether it learns some requested token in round $r$ or not.
	However, $v$ knows the token requests it sent over its adjacent edges in round $r-1$. 
	Then, by knowing the adjacent edges in round $r$, and the fact that complete nodes immediately respond to requests, $v$ knows what tokens it learns in round $r$. \onlyLong{The pseudocode is given in Algorithm~\ref{alg:single-source-unicast}.}

		\onlyLong{
\begin{algorithm}[t]
\caption{\sc Single-Source-Unicast}\label{alg:single-source-unicast}
Initially, the source node labels the tokens from $1$ to $k$ as token IDs, and the following code is run by any node $v$ in any round $r$.
\begin{algorithmic}[1]
\If {$v$ is complete}
	\ForAll {$v$'s neighbor $u$}
		\If {$u$ does not know $v$'s completeness}
			\State send \textit{Completeness} to $u$
		\ElsIf {$u$ sent \textit{Request($i$)} in round $r-1$}
			\State send the $i^{th}$ token to $u$
		\EndIf
	\EndFor
\ElsIf{$\{b_1, b_2, \dots, b_\gamma\}$ is the ID set of missing tokens for $v$}
	\State $j \gets 0$
	\ForAll {$v$'s new edge $e$}
		\If {$j<\gamma$} 
			\State $j \gets j+1$
			\State send \textit{Request($b_j$)} over $e$
		\EndIf
	\EndFor
	\ForAll {$v$'s idle edge $e$}
		\If {$j<\gamma$} 
			\State $j \gets j+1$
			\State send \textit{Request($b_j$)} over $e$
		\EndIf
	\EndFor
	\ForAll {$v$'s contributive edge $e$}
		\If {$j<\gamma$} 
			\State $j \gets j+1$
			\State send \textit{Request($b_j$)} over $e$
		\EndIf
	\EndFor
\EndIf 
\end{algorithmic}
\end{algorithm}
}

\onlyShort{\vspace{-0.1in}}
\subsubsection{Analysis}\label{sec:adv-pays-single-source-analysis} 
\onlyShort{\vspace{-0.1in}} 
	First, let us argue the message complexity of the algorithm. 
	Then, we show that with a natural stability assumption the time complexity is also small.
\vspace{.1in}	
\begin{theorem}\label{thm:unicast-msg}
	Given $k$ tokens to disseminate in a dynamic network against a strongly adaptive adversary, the Single-Source Unicast Algorithm has 1-adversary-competitive message complexity of $O(n^2+nk)$.  
\end{theorem}

	\begin{proof}
	\vspace{-0.1in}
		There are three different types of messages sent by nodes during the algorithm execution; (1) token, (2) completeness announcement, and (3) token request. 
		Each node sends the request of each distinct token to at most one neighbor in a round. 
		If the connection to that complete neighbor remains for the very next round, then the requested token will be successfully received by the node and the node stops sending this token request. 
		Therefore, each distinct token is received by each node once, and hence there are at most $O(nk)$ sent messages of type 1 throughout the execution. 
		
		Each of the $n$ nodes informs at most $n-1$ other nodes about its completeness throughout the execution. 
		Since each node avoids informing the same node more than once, at most $O(n^2)$ messages of type 2 are sent throughout the execution. 
		    
		It remains to show that the number of sent messages of type 3 is at most $O(nk)+\TC(\calE)$ during execution $\calE$. 
		In each round where a token request is sent by some node, a new token is received in the next round unless the edge is removed. 
		Therefore, we can say that the number of token requests sent at any time is at most $O(nk)$ plus the number of edge deletions. 
		$O(nk)$ comes from the fact that there exist $k$ tokens and each token is received by at most $O(n)$ nodes, each token once. 
		Furthermore, since we assume that the initial graph is an empty graph, the number of edge deletion is upper bounded by $\TC(\calE)$. 
	\end{proof}
	
	
	In the following, we argue that with a natural stability assumption, the algorithm disseminates all the tokens and terminates fast. 
	The following two lemmas show that prioritization of sending token requests over different edge types ensures fast dissemination.
	
	\begin{definition}[Futile Round]
		Round $r$ is a \textit{futile round}, if no token request is sent over a contributive edge in round $r$, and no token learning occurs in rounds $r+1$ and $r+2$.
	\end{definition}
	\begin{lemma}\label{lem:idle-edge}
		Let $r$ be an arbitrary futile round in any execution of the Single-Source Unicast Algorithm on a $3$-edge stable dynamic network. 
		Then, if there exist $\ell$ bridge nodes in round $r$, at least $\ell$ idle edges are removed at the end of round $r$.  
	\end{lemma}
	
	\begin{proof}
		First, let us show that every bridge node has an adjacent idle edge in round $r$. 
		If there exists a new edge in round $r$, due to the 3-edge stability property and the higher priority of sending requests on new edges, a token is learned in at least one of rounds $r+1$ or $r+2$.
		Hence, there exists no new edge in round $r$. 
		Now, for the sake of contradiction, let us assume that there exists a bridge node $b$ in round $r$ that does not have an adjacent idle edge. 
		Since $b$ cannot have an adjacent new edge either, it must have at least one contributive edge.
		Therefore, $b$ sends a request over at least one of its contributive edges in round $r$, contradicting the assumption that $r$ is a futile round. 
		
		Since every bridge node has an idle edge and no new edge, due to the mentioned priority rules, a bridge node sends a request over at least one of its idle edges. 
		Since no new token is learned in round $r+1$, the idle edge carrying a request must be removed. 
		Hence, from each bridge node at least one idle edge is removed at the end of round $r$. 
	\end{proof}
	
	\begin{lemma}\label{lem:futile-round}
		  In any execution of the Single-Source Unicast Algorithm on a 3-edge stable $n$-node dynamic network, there are at most $n$ futile rounds until the last token request is sent.
	\end{lemma}
	
	\begin{proof}
	Let us first argue that it is not possible for a new edge to become idle. 
	For any round $r>0$, consider an arbitrary new edge $e=\{u,v\}\in E^+_r$, where $u$ is complete and $v$ is incomplete. 
	Then in round $r+2$, either $e$ is contributive or $v$ is complete.
	Because, the only case that $v$ does not send a token request over $e$ in rounds $r$ or $r+1$ is when $v$ sends all its left token requests over its other new edges in rounds $r$ or $r+1$.
	Then, due to $3$-edge stability, $v$ will receive its requested tokens by the end of round $r+2$ and becomes complete. 
	Otherwise, $v$ sends a token request over $e$ in rounds $r$ or $r+1$, and hence $e$ becomes contributive by the end of round $r+2$.  
	
	Then, the only case when an edge becomes idle in round $r$, is when both endpoints are incomplete in round $r-1$ and only one of  them becomes complete in round $r$.  
	Since each node $v$ becomes complete only once, the number of $v$'s idle edges never increases throughout the execution after $v$'s completion.

		Now consider an arbitrary futile round where the largest number of idle edges of any complete node in a futile round is $x$. 
		Hence, there exist at least $x$ bridge nodes in that round. 
		Thus, by Lemma \ref{lem:idle-edge}, at least $x$ idle edges are removed at the end of that futile round. 
		As a result, one can see that there cannot be any idle edges, and hence any futile rounds, after having $n$ futile rounds. 
		This shows that the number of futile rounds is at most $n$ until the last token request is sent. 
	\end{proof} 
	
	\begin{theorem}\label{thm:unicast1-time}
		Given $k$ tokens to disseminate, if the dynamic graph is $3$-edge stable, the Single-Source Unicast Algorithm terminates in $O(nk)$ rounds and all the nodes  receive all the $k$ tokens. 
	\end{theorem}
	\begin{proof}
		Consider any time $t$ during an arbitrary execution of the Single-Source Algorithm that is not terminated yet.  
		Let $k'$ denote the number of token learnings in $[0,t]$.
		Let us show that the number of periods of two consecutive rounds in $[1,t]$ in which no token is learned is at most $k'+n$.
		This leads to $O(nk)$ running time for the algorithm. 
		
		Let $r$ and $r+1$ be arbitrary two consecutive rounds in $[1,t]$, where no token is learned. 
		Hence, there is no new edge in round $r-1$, otherwise, a token would have been learned in round $r$ or $r+1$ due to the 3-edge stability property and the higher priority of sending token requests on new edges. 
		Then, there are two possibilities:
		\begin{itemize}
			\item{Case 1:}
				At least one contributive edge carries a token request in round $r-1$.
				Since it is assumed that no token is learned in round $r$, the edge must be removed by the adversary at the end of round $r-1$. 
				Therefore, we can map one of the removed contributive edges to round $r$. 
				Doing so, for any such round $r$, a distinct token learning in $[0,t]$ is mapped to $r$ (i.e., one of the token learnings that happened on the removed contributive edge after its last insertion). 
				Therefore, since there is a one to one mapping between such rounds and a subset of token learnings in $[0,t]$, the number of such rounds (i.e., $r$) is not more than the number of token learnings in $[0,t]$. 
			\item{Case 2:}
				No contributive edge carries a token request in round $r-1$. 
				Therefore, round $r-1$ is a futile round. 
				Then, based on Lemma \ref{lem:futile-round}, the number of such rounds (i.e., round $r$) is at most $n$ throughout the execution. 
		\end{itemize}
	\end{proof}
	
 \subsection{Multiple Source Nodes}
\label{multi-source}
	Let us consider a more general case where the tokens are initially given to more than one source node. 
	Assume that there are $s$ source nodes $a_1< a_2< \dots < a_s$ such that for $1\leq i\leq s$, $a_i$ is initially given $k_i$ tokens. 
	Hence, in total $k = \sum_i k_i$ tokens  need to be disseminated.

\subsubsection{Strongly Adaptive Adversary}\label{sec:adaptive}
 \onlyShort{\vspace{-0.1in}}
	To solve this problem against a strongly adaptive adversary, we present a deterministic algorithm with $O(n^2s+nk)+\TC(\calE)$ message complexity. 
	It extends the Single-Source Unicast Algorithm, and has the same running time if the network has the same stability assumption (i.e., 3-edge stability). 
	However, it has a higher message complexity than the Single Source Unicast Algorithm since each node needs to announce its completeness regarding $s$ different source nodes to other nodes in its neighborhood throughout the algorithm execution.  
	
	Since there are more than one source nodes, we need to include the intended source node in the definitions of Section \ref{sec:single-source}. 
	So we say a node is complete \textit{with respect to} source node $a$, if it has received all the tokens originated at $a$.
	Similarly, a node is called a bridge node \textit{with respect to} source node $a$, if it is an incomplete node with respect to $a$ and is connected to a node which is complete with respect to $a$. 
\\[.1cm]
\noindent
\textbf{\textit{Multi-Source-Unicast Algorithm}}\label{alg:unicast2}\\
	The algorithm considers a priority over the dissemination of tokens from different sources.
	To do so, in each round, all nodes give the highest priority to the dissemination of the tokens from the minimum known source node whose dissemination is not yet complete. 
	In the sequel, we explain the details of implementing this idea.	
	
	Initially, each source node $x$ considers an arbitrary order of its tokens and assigns a token identifier containing its own ID and an integer $i$ (i.e., $\langle ID_x,i\rangle$) to its $i^{th}$ token. 
	Moreover, we assume that each source node becomes complete with respect to itself at time $0$.	
	To avoid redundant communication, each node $v$ keeps some information about the execution history by constantly updating the following sets. 
	$R_v(x)$ is the set of all nodes that are informed by $v$ about the $v$'s completeness with respect to $x$. 
	$S_v(x)$ is the set of nodes that informed $v$ about their completeness with respect to $x$. 
	$I_v$ is the set of all source nodes with respect to which $v$ is complete. 
	Then each node $v$ in each round of the execution does the following three tasks in parallel: 
		(1) For each edge $\{v,w\}$, if there is any source node $x$ such that $x\in I_v$ and $w\notin R_v(x)$, it picks the minimum such $x$ and sends ``completeness announcement with respect to $x$'' to $w$;
		(2)  For each edge $\{v,w\}$, if $v$ received a request for token $\tau$ from $w$ in the previous round, then it sends $\tau$ to $w$; (3) Node $v$ picks the minimum $x$ such that $x\notin I(v)$ and $S_v(x)\neq \emptyset$. 
		Then, regarding sending token requests, it acts similarly to the Single-Source Unicast Algorithm as there exists only the single source $x$ in the network.

\onlyShort{The proof of the following Theorem is similar to the  Proof of Theorem \ref{thm:unicast-msg}, i.e., the single-source case.}

	\begin{theorem} \label{thm:unicast-multi1-msg}
		To disseminate $k$ tokens which are initially distributed among $s$ source nodes, Multi-Source Unicast Algorithm has a 1-adversary-competitive message complexity of $O(n^2s+nk)$.  
	\end{theorem}
	\onlyLong{
	\begin{proof}\label{prf:unicast2-msg-complexity}
	Arguing the message complexity of Multi-Source Unicast Algorithm is almost similar to the proof of Theorem \ref{thm:unicast-msg}.
	Similarly, we consider the three different types of messages throughout the algorithm execution; (1) token, (2) completeness announcement, and (3) token request.
	The number of tokens of type 1 and 3 is exactly the same as running the Single Source Unicast Algorithm. 
	However, the number of messages of type 2 differs. 
	In case of running the Single Source Unicast Algorithm, each node needs to inform any other node in its neighborhood about its completeness once throughout the algorithm execution.
	The reason is that there is only one source node, and each node achieves completeness just regarding the only source node in the network. 
	But in case of running Multi-Source Unicast Algorithm, each node becomes complete regarding $s$ different source nodes.
	Therefore, each node should announce its completeness regarding each of the $s$ source nodes to every other node in its neighborhood throughout the algorithm execution, which leads to $O(n^2s)$ messages in total. 
	As a result, $O(nk)$ messages of type 1, $O(n^2s)$ messages of type 2, and $O(nk)+\TC(\calE)$ messages of type 3 proves the 1-adversary-competitive message complexity of $O(n^2s+nk)$ for Multi-Source Unicast Algorithm. 
\end{proof}
	}
	
	\onlyLong{
	\begin{theorem}\label{thm:unicast-multi1-time}
		Given $k$ tokens to disseminate, if the dynamic graph is 3-edge stable Multi-Source Unicast Algorithm terminates in $O(nk)$ rounds and all the nodes have received all the $k$ tokens.
	\end{theorem}
	
	\begin{proof}
		Theorem \ref{thm:unicast1-time} states when all the $k$ tokens are initially given to one source node, by running Single-Source Unicast Algorithm, $k$-token dissemination is complete in at most $O(nk)$ rounds. 
		Multi-Source Unicast Algorithm guarantees that the minimum ID source node that its token dissemination is not complete yet runs the Single-Source Unicast Algorithm without any interference until its token dissemination is complete. 
		It is guaranteed by having all the nodes giving the highest priority to the token dissemination of the the minimum ID source node with incomplete token dissemination. 
		
		Therefore, if the Single-Source Unicast Algorithm solves $k$-token dissemination in $cnk$ rounds for some constant $c$, then the token dissemination of the first minimum ID source node is complete after $cnk_1$ rounds and the second one after the next $cnk_2$ rounds and so on. 
		Hence, the whole running time is $O(nk)$, where $k = \sum_{i=1}^s k_i$.

	\end{proof}
	}

 \onlyShort{\vspace{-0.3in}}
\subsubsection{Oblivious Adversary}
\label{sec:oblivious-adv}	
 \onlyShort{\vspace{-0.1in}}
	In case the ratio of the number of disseminated tokens to the number of source nodes is large enough, i.e., $k/s=\Omega(n)$, the algorithm presented in Section~\ref{sec:adaptive} has an efficient linear amortized message complexity.
	However, for example, in case of having $\Omega(n)$ source nodes and $O(n)$ tokens to be disseminated, the amortized message complexity of the algorithm would be $\Omega(n^2)$ due to Theorem~\ref{thm:unicast-multi1-msg}. In this section, we focus on instances with large number of source nodes and $o(n^2)$ tokens in total are distributed arbitrarily among the source nodes.
	Assume that the number of source nodes and the total number of tokens are initially known to the nodes.
	Then, we show that by weakening the adversary from an adaptive one to an oblivious one, a better amortized message complexity can be achieved when the ratio $k/s$ is small.
	Hence in the sequel we assume that $k/s=o(n)$ and $k=o(n^2)$.

	The key idea is to efficiently reduce the number of source nodes and then simply run the Multi-Source-Unicast algorithm for this smaller set of sources. 
	Hence, the algorithm runs in two phases. 
	In the first phase, a (small) subset of nodes is chosen as new source nodes, and all the tokens are efficiently sent to these new source nodes. Let us call the new source nodes \textit{centers}. 
	Then, in the second phase, the Multi-Source-Unicast algorithm is executed with the centers as the source nodes. 
	 
	Let us now explain the first phase in details. 
	If the number of source nodes is less than $n^{2/3}\log^{5/3} n$, nothing is done in the first phase and the second phase is started right away by running the Multi-Source-Unicast algorithm (by considering all the source nodes as centers). 
	Therefore, in the sequel, let us assume that the number $s$ of source nodes is more than $n^{2/3}\log^{5/3} n$. 
	We aim to reduce the number of source nodes from $s$ to $f$, where parameter $f$ denoting the number of centers will be determined later. 
	Then, the $f$ centers own all the tokens at the end of the first phase.
	
	Each node independently marks itself as a center with probability $f/n$. 
	Therefore, in expectation, there are $f$ centers. 
	Then, each token owned by any source node (which is not marked as a center) needs to reach to some center. 
	The tokens owned by one source node may reach different centers.
	However, each token is owned by exactly one center at the end of the first phase. 
	To have this new token assignment, each of these tokens performs a random walk (in parallel) until they reach a center. 
	Once a token reaches a center, it stops there and the center owns the token.  	
	Since in expectation, there are $f$ uniformly random centers among the $n$-nodes, any fixed set of $O(n \log n/f )$ distinct nodes must have at least one center with high probability (w.h.p.). That is, each random walk token has to visit $\Omega(n\log n/f)$ distinct nodes to guarantee that it hits a center w.h.p.
	For this, we apply a known random walk visit bound (see Lemma~\ref{lem:lyon-dynamic} below) for the dynamic setting \cite{sarmaMP15}.  
	
	To perform the desired random walks, we construct a virtual $n$-regular multigraph by adding an appropriate number of self-loops to the network at each round. 
	To do so, for any round $r$, each node with degree $\delta$ in the graph adds $n-\delta$ virtual self-loops as its adjacent virtual edges. 
	Note that a random walk step on a virtual edge is not count in the message complexity, but it increases the time complexity. 
	Due to the assumed bandwidth restriction (i.e., congestion) of the actual edges, not necessarily all the tokens perform a random walk step in each round.  
	Therefore, we say a token is \textit{active} in a round when it performs a random walk step whether it traverses an actual or virtual edge. 
	Otherwise, we say that the token is \textit{passive}.
	Consider $\gamma=(n\log n)/f$ as a predefined degree threshold. 
	We call a node with degree larger than $\gamma$ a \textit{high-degree} node; otherwise it's a \textit{low-degree} node. 
	Recall that a high-degree node must have at least one center among its neighbors with high probability.
	\onlyLong{
\begin{algorithm}[t]
\caption{\sc Oblivious-Multi-Source-Unicast}\label{alg:oblivious-unicast}
\textbf{Input to each node:} Number of source nodes $s$ and total number of tokens $k$\\
\textbf{Output:} Every node receive all the $k$ tokens \\
\begin{algorithmic}[1]
	\If {$s \leq n^{2/3\log^{5/3} n}$}
		\State Run {\sc Multi-Source-Unicast} algorithm with the $s$ source nodes \label{leb:multi-source-algo}
	\ElsIf{$s > n^{2/3}\log^{5/3} n$} \Comment{[Phase~1: Reducing no. of source nodes to $f = n^{1/2} k^{1/4}\log^{5/4} n$ centers]}
		\State Each node elects and marks itself as a center with probability $f/n$  
		\For{round $r = 1, 2, \dots \ell$} \label{leb:iteration} \Comment{[$\ell = k^{\frac{1}{4}}n^{\frac{5}{2}}\log^{\frac{9}{4}} n$]}
			
			\State Each node $u$ owning at least one token does the following for each token $\tau$:
			\If{$d(u) < n^{1/2}(k\log n)^{-1/4}$} \Comment{[low degree; $d(u)$ is degree of $u$ in round $r$]} 
				\State With probability $1/d(u)$, go to Step $9$, and otherwise Step $10$ 
				\State Send $\tau$ to a random neighbor  \Comment{[If congestion allows, otherwise keep the token]}
			\ElsIf{$d(u) \geq n^{1/2}(k\log n)^{-1/4}$} \Comment{[high degree]} 
			\State Send one token (if any) to each of the neighboring centers 
			\EndIf
		\EndFor
		\State Go to Step~\ref{leb:multi-source-algo} with $s =f$ \Comment{[Phase~2: Run {\sc Multi-Source-Unicast} algorithm]}
	\EndIf
\end{algorithmic}
\end{algorithm}
}

	Consider an arbitrary low-degree node $v$ with degree $\delta_v$, and let $T$ be the set of tokens at node $v$ at the beginning of round $r$.
	Node $v$ processes each token $\tau$ in $T$ as follows. 
	With probability $1-\delta_v/n$, token $\tau$ traverses a self-loop, i.e., it remains at node $v$. 
	With probability $\delta_v/n$, $v$ chooses one of its adjacent edges $e$ uniformly at random, and if $v$ has not yet sent any token over $e$ in round $r$, token $\tau$ is sent over $e$.
	Therefore, a token at a low-degree node might be passive in a round because of the congestion for the edges. 
	Now consider a high-degree node $u$ with degree $\delta_u$ in round $r$. 
	Then w.h.p. node $u$ has at least one center among its neighbors. 
	To each of its neighboring centers, $u$ sends one of the tokens owned by node $u$ (if any) at the beginning of round $r$. 
	Since the number of $u$'s neighboring centers might be less than the number of tokens at node $u$, not necessarily all the tokens at node $u$ are sent to the neighboring centers in the round $r$. 
	Therefore, a token at $u$ is passive until it is either sent to one of $u$'s neighboring centers, or the degree of $u$ becomes lower than the threshold and the token resumes the random walk. This way a token continues walking until it reaches a center. \onlyLong{The pseudocode is given in Algorithm~\ref{alg:oblivious-unicast}.}

\noindent {\bf Analysis.}
	Consider the random walk of an arbitrary token $\tau$ in the given dynamic graph $G$. 
	As explained in the algorithm description, token $\tau$ is not necessarily active in all rounds throughout the algorithm execution.
	Let $G_\tau$ denote the (not necessarily consecutive) subsequence of $G$ such that $\tau$ is active in each and every graph in $G_\tau$. 
	In each graph in $G_\tau$ (except the last one), token $\tau$ is sent from a node $u$ to a node $v$ such that $u$ is a low-degree node.
	Therefore, all the nodes visited by $\tau$ in $G_\tau$ have actual degree at most $\gamma$.
		
\onlyLong{\vspace{.2cm}}		
\begin{lemma}[Lemma~6.7 in \cite{sarmaMP15}]\label{lem:lyon-dynamic}
Let $\mathcal{G}$ be a $d$-regular dynamic graph controlled by an oblivious adversary. Let $N_x^t(y)$ denote the number of visits of a random walk to vertex $y$ by time $t$, given that the random walk started at node $x$. $N_x^t(y)$ could be zero or a positive number. Then for any nodes $x$, $y$ and for all $t = O(\tau_{mix})$, where $\tau_{mix}$ is the (dynamic) mixing time of $\mathcal{G}$, 
$\,\Pr \left(N_x^t(y) \geq 2^{c+3}\cdot d\sqrt{t + 1} \log n \right) \leq 1/n^c$, for any constant $c$.
\end{lemma}
\onlyLong{\vspace{.2cm}}		

	The above lemma holds for any random walk with an arbitrary graph sequence provided by an oblivious adversary. 
	We refer to \cite{sarmaMP15} for more details.
	It states that a random walk of length $L$ on a $d$-regular dynamic graph visits at least $L/(2^{c+3}d\sqrt{L + 1} \log n)$ i.e., $\Omega(\sqrt{L}/d\log n)$ distinct nodes with high probability (for $c = 4$). 
	Since only token traversal over the actual edges increases the message complexity, regarding Lemma~\ref{lem:lyon-dynamic}, (to analyze the worst case message complexity) we only consider the upper bound for the actual degree of all the visited nodes by $\tau$, which is $\gamma$.
	 To have $\tau$ performing $L$ actual steps, the walk takes at least $\Theta(nL/\gamma)$ steps w.h.p. on the constructed $n$-regular multigraph (using standard Chernoff bound). 
	 Therefore, due to Lemma~\ref{lem:lyon-dynamic}, $\tau$ visits $\Omega\left((\sqrt{nL/\gamma})/n\log n\right) = \Omega\left(\sqrt{L/(\gamma n \log^2 n)}\right)$ distinct nodes.
	 As we discussed earlier, to have $\tau$ visiting a center during its walk w.h.p., it is enough that $\tau$ visits at least $(n\log n)/f$ distinct nodes. 
	 Thus, we get $L = \Omega\left((n^4\log^5 n) /f^3\right)$,  by setting $\left(\sqrt{L/(\gamma n \log^2 n)}\right) \geq (n\log n)/f$ and $\gamma = (n\log n)/f$. This implies that each token performs a random walk of length at least $(n^4\log^5 n)/f^3$ to guarantee that it reaches a center w.h.p. Since this is true for an arbitrary random walk token w.h.p, by union bound, it is also true for all the tokens.

	 The following theorem \onlyShort{ (proof deferred to the full paper \cite{AKKMP18}) }shows that by setting the parameters properly, the desired message complexity is achieved.
\begin{theorem}\label{app:thm:main-msg-oblivious}
There is an algorithm with message complexity $O(n^{5/2}k^{1/4}\log^{\frac{5}{4}} n)$ to disseminate $k = o(n^2)$ tokens from $\Omega(n^{2/3}\log^{5/3} n)$ source nodes in a dynamic network, in which the topology is controlled by an oblivious adversary. Hence, the amortized message complexity of the algorithm is $O((n^{5/2}\log^{\frac{5}{4}} n)/k^{3/4})$.
\end{theorem}
\onlyLong{
\begin{proof}
	In the first phase, at most $k$ tokens perform random walks of $L$ (actual) steps each to reach some center. 
	Note that this excludes message cost for the self-loop (virtual) edges. Therefore, it costs $kL$ messages in the first phase. 
	In the second phase, we run Multi-Source-Unicast algorithm with $f$ source nodes. 
	Due to Theorem~\ref{thm:unicast-multi1-msg}, therefore, the message complexity of the second phase is $O(fn^2 + nk)$. 
	Thus, the total message complexity is $O(kL + fn^2 + nk)$. 
	Parameter $f$ is sub-linear in $n$, and $L = \Omega\left((n^4\log^5 n\right) /f^3)$.
	Hence, $L$ is larger than $n$, and consequently $kL > kn$. 
	The message complexity is $O(kL + fn^2)$. 
	To fix parameter $f$, let us optimizing the sum $(kL + fn^2)$ as follows.
\begin{align*}
 & kL = fn^2 \\
 \Rightarrow & \, L = fn^2/k \\ 
 \Rightarrow  &\, n^4\log^5 n/f^3 = fn^2/k  \hspace{.7cm} [\text{Substituting $L = (n^4\log^5 n) /f^3$}]\\
 \Rightarrow & \, f = n^{1/2} k^{1/4}\log^{5/4} n
\end{align*}
Thus, the total message complexity is $O(fn^2) = O(n^{\frac{5}{2}} k^{\frac{1}{4}}\log^{\frac{5}{4}} n)$. 

Therefore, the amortized message complexity to disseminate $k$ tokens is 
$$O(n^{\frac{5}{2}} k^{\frac{1}{4}}\log^{\frac{5}{4}} n)/k = O\left(\frac{n^{\frac{5}{2}}\log^{\frac{5}{4}} n}{k^{\frac{3}{4}}}\right).$$ 
\end{proof}
}
The following table highlights the amortized message cost for different sizes of the token set. Recall that, by our assumption $s \geq n^{2/3}\log^{5/3} n$ and $k=o(n^2)$, and $k\geq s$ always. 
\begin{table}[h]
\footnotesize
\centering
  \begin{tabular}{ | >{\centering\arraybackslash} m{6.5cm} | >{\centering\arraybackslash} m{6.5cm} |}
    \hline
    Number of disseminated tokens ($k$) & Amortized message complexity \\ \hline \hline
    \rule{0pt}{11pt} $O(n^{\frac{2}{3}}\log^{\frac{5}{3}} n)$ &  $O(n^2)$ \\
    \hline
     \rule{0pt}{11pt} $O(n)$ & $ O(n^{\frac{7}{4}}\log^{\frac{5}{4}} n)=o(n^2)$\\
    \hline
     \rule{0pt}{11pt} $O(n^{\frac{3}{2}})$ & $O(n^{\frac{11}{8}}\log^{\frac{5}{4}} n)$ \\
    \hline
     \rule{0pt}{11pt} $O(n^2)$ & $O(n\log^{\frac{5}{4}} n)$\\
    \hline
  \end{tabular}
  \caption{The amortized message complexity for different number of tokens.}
  \label{tab:results}
\end{table}

\noindent {\bf Remark.}
	As mentioned before, in case of having less than $n^{\frac{2}{3}}\log^{\frac{5}{3}} n$ source nodes, Multi-Source-Unicast algorithm is executed.
	It is a deterministic algorithm, and hence works properly against an oblivious adversary. 
	The total message cost of Multi-Source-Unicast Algorithm is $O(n^2s + nk)$  (cf. Theorem~\ref{thm:unicast-multi1-msg}).
	Therefore, the amortized message complexity is $ O(\frac{n^{2}s}{k} + n)$, which is upper bounded by $O(n^2)$, since the number of tokens is always larger than the number of source nodes, i.e., $s/k \leq 1$.
	Therefore, when the number of source nodes is less than $n^{\frac{2}{3}}\log^{\frac{5}{3}} n$, Multi-Source-Unicast algorithm is more efficient.
\onlyShort{In the full paper \cite{AKKMP18}, we also analyze the running time of the algorithm.}

\onlyLong{
	Now let us analyze the running time of the algorithm. 
	Since there are total $k = o(n^2)$ tokens and at least $s = n^{2/3}\log^{5/3} n$ source nodes, a source node may have as many as $O(k-s)$ tokens to disseminate in the beginning. 
	Further, since the dynamic graph is $n$-regular, as many as $O(n)$ tokens from each node can be executed in parallel with at most $O(\log n)$ congestion over an edge. The reason is that if each node starts $O(n)$ random walks in parallel, in expectation, each edge carries at most $2$ walks (from both ends) in each round, and hence there will be at most $O(\log n)$ congestion over an edge with high probability. 
	Therefore, to perform $O(k-s)$ random walks (corresponding to $O(k-s)$ tokens from a source node) in parallel, there would be at most $O((k-s)\log n/n)$ delay per step w.h.p. 
	Another reason for a delay in the random walk of a token is that the token is at a high-degree node in some round and the number of neighboring centers is less than the number of tokens at that node in that round. 
	Note that the number of such rounds is at most $k$, since in each such (delay) round there is at least one token that is being sent to a center. 
	
	Since the length of the random walks (including virtual steps\footnote{The virtual steps are counted towards running time of the algorithm.}) is $O(nL)$ (assuming the worst case actual degree $O(1)$ for the running time), the total time of the first phase is $O((k-s)\log n/n \cdot (nL)+k)$ rounds. 
	Since the second phase is the execution of Multi-Source-Unicast algorithm, it takes $O(nk)$ time with the additional natural condition that the dynamic graph is $3$-edge stable, as follows from Theorem~\ref{thm:unicast-multi1-time}. 
	Hence, the total running time in phase~1 and phase~2 is $O((k-s)L\log n + k + nk)$ rounds. 
	The time bound becomes $O\Big((k-s)n^{\frac{5}{2}}\log^{\frac{9}{4}}n/k^{\frac{3}{4}} + nk\Big) \leq O\Big(k^{\frac{1}{4}}n^{\frac{5}{2}}\log^{\frac{9}{4}} n \Big)$, as $L = O\Big(n^{\frac{5}{2}}\log^{\frac{5}{4}}n/k^{\frac{3}{4}}\Big)$ and $k=o(n^2)$.
	}



\vspace{-0.15in}
\section{Conclusion and Open Problems}
\label{sec:conclusion}
\onlyShort{\vspace{-0.1in}}

We studied the message complexity of information spreading in dynamic networks.
\onlyLong{While time complexity has been studied more intensely, understanding the message complexity
 in various dynamic network models is likely to shed light on the time complexity as well.}
Several open questions arise from our work. One key question is that we do note have tight bounds
on the amortized message complexity of unicast under the strongly adaptive adversary (when not charging the adversary for topological changes). The only known bounds are the trivial $O(n^3)$ upper and $\Omega(n)$ lower bounds.

A contribution of our work is introducing the adversary-competitive message complexity which is useful for studying algorithmic
performance in dynamic networks as a function of the dynamism. We were able to show an optimal amortized message bound
for unicast in this model for both the single-source and multi-source setting, when the number of tokens is large. However,
when the number of tokens is small (say $n$) and they start from multiple sources (an important special case is
one token starts from each node), we do not have a good bound. We were able to show only a $o(n^2)$ amortized bound
under a weaker (oblivious) adversary. Improving this bound for oblivious adversary  is an interesting open problem or showing a non-trivial bound for the strongly adaptive adversary is an interesting open problem. \onlyLong{In the case of oblivious adversary, we assumed the number of source nodes and the number of tokens as inputs. It would nice if one can try to relax the assumptions.    
 Also, developing efficient protocols for dynamic networks
that perform well under the adversary-competitive measure for various problems is an interesting research goal.}

\vspace{1cm}
\bibliographystyle{abbrv}
\bibliography{tokenforwarding,sym_diff_sampling,refs-dynamic,dynamic}

\def\cprime{$'$}
\begin{thebibliography}{10}

\bibitem{afek+ag:dynamic}
Y.~Afek, B.~Awerbuch, and E.~Gafni.
\newblock Applying static network protocols to dynamic networks.
\newblock In {\em IEEE FOCS}, pages 358--370, 1987.

\bibitem{ahlswede+cly:coding}
R.~Ahlswede, N.~Cai, S.~Li, and R.~Yeung.
\newblock Network information flow.
\newblock {\em Transactions on Information Theory}, 46(4):1204--1216, 2000.

\bibitem{alon+blp:radio}
N.~Alon, A.~Bar-Noy, N.~Linial, and D.~Peleg.
\newblock A lower bound for radio broadcast.
\newblock {\em Computer and System Sciences}, 43(2):290--298, 1991.

\bibitem{attiya+w:distributed}
H.~Attiya and J.~Welch.
\newblock {\em Distributed Computing: Fundamentals, Simulations and Advanced
  Topics}.
\newblock John Wiley Interscience, 2004.

\bibitem{gopal-disc16}
J.~Augustine, C.~Avin, M.~Liaee, G.~Pandurangan, and R.~Rajaraman.
\newblock Information spreading in dynamic networks under oblivious
  adversaries.
\newblock In {\em DISC}, pages 399--413, 2016.

\bibitem{podc13}
J.~Augustine, G.~Pandurangan, and P.~Robinson.
\newblock Fast byzantine agreement in dynamic networks.
\newblock In {\em PODC}, pages 74--83, 2013.

\bibitem{disc15}
J.~Augustine, G.~Pandurangan, and P.~Robinson.
\newblock Fast byzantine leader election in dynamic networks.
\newblock In {\em DISC}, pages 276--291, 2015.

\bibitem{avin08}
C.~Avin, M.~Kouck\'{y}, and Z.~Lotker.
\newblock How to explore a fast-changing world (cover time of a simple random
  walk on evolving graphs).
\newblock In {\em ICALP}, pages 121--132, 2008.

\bibitem{awerbuch+bbs:route}
B.~Awerbuch, P.~Berenbrink, A.~Brinkmann, and C.~Scheideler.
\newblock Simple routing strategies for adversarial systems.
\newblock In {\em IEEE FOCS}, pages 158--167, 2001.

\bibitem{awerbuch+bs:anycast}
B.~Awerbuch, A.~Brinkmann, and C.~Scheideler.
\newblock Anycasting in adversarial systems: Routing and admission control.
\newblock In {\em ICALP}, pages 1153--1168, 2003.

\bibitem{awerbuch+l:flow}
B.~Awerbuch and F.~T. Leighton.
\newblock Improved approximation algorithms for the multi-commodity flow
  problem and local competitive routing in dynamic networks.
\newblock In {\em ACM STOC}, pages 487--496, 1994.

\bibitem{bar-noy+gns:multicast}
A.~Bar-Noy, S.~Guha, J.~Naor, and B.~Schieber.
\newblock Message multicasting in heterogeneous networks.
\newblock {\em SIAM Journal on Computing}, 30(2):347--358, 2000.

\bibitem{bar-yehuda+gi:radio}
R.~Bar-Yehuda, O.~Goldreich, and A.~Itai.
\newblock On the time-complexity of broadcast in radio networks: an exponential
  gap between determinism and randomization.
\newblock In {\em ACM PODC}, pages 98--108, 1987.

\bibitem{BCF11}
H.~Baumann, P.~Crescenzi, and P.~Fraigniaud.
\newblock Parsimonious flooding in dynamic graphs.
\newblock {\em Distributed Computing}, 24(1):31--44, 2011.

\bibitem{saia}
M.~A. Bender, J.~T. Fineman, M.~Movahedi, J.~Saia, V.~Dani, S.~Gilbert,
  S.~Pettie, and M.~Young.
\newblock Resource-competitive algorithms.
\newblock {\em {SIGACT} News}, 46(3):57--71, 2015.

\bibitem{santoro}
A.~Casteigts, P.~Flocchini, W.~Quattrociocchi, and N.~Santoro.
\newblock Time-varying graphs and dynamic networks.
\newblock {\em CoRR}, abs/1012.0009, 2010.
\newblock Short version in ADHOC-NOW 2011.

\bibitem{clementi+ms:radio}
A.~Clementi, A.~Monti, and R.~Silvestri.
\newblock Distributed multi-broadcast in unknown radio networks.
\newblock In {\em ACM PODC}, pages 255--264, 2001.

\bibitem{CCDFIPPS13}
A.~E.~F. Clementi, P.~Crescenzi, C.~Doerr, P.~Fraigniaud, M.~Isopi,
  A.~Panconesi, F.~Pasquale, and R.~Silvestri.
\newblock Rumor spreading in random evolving graphs.
\newblock In {\em Proc. of 21st Annual European Symposium on Algorithms (ESA)},
  pages 325--336, 2013.

\bibitem{CMMPS08}
A.~E.~F. Clementi, C.~Macci, A.~Monti, F.~Pasquale, and R.~Silvestri.
\newblock Flooding time in edge-markovian dynamic graphs.
\newblock In {\em Proc. of the 27th Annual {ACM} Symposium on Principles of
  Distributed Computing (PODC)}, pages 213--222, 2008.

\bibitem{CPMS07}
A.~E.~F. Clementi, F.~Pasquale, A.~Monti, and R.~Silvestri.
\newblock Communication in dynamic radio networks.
\newblock In {\em Proc. of the 26th Annual {ACM} Symposium on Principles of
  Distributed Computing (PODC)}, pages 205--214, 2007.

\bibitem{CS15}
A.~E.~F. Clementi and R.~Silvestri.
\newblock Parsimonious flooding in geometric random-walks.
\newblock {\em J. Comput. Syst. Sci.}, 81(1):219--233, 2015.

\bibitem{sarmaMP15}
A.~{Das Sarma}, A.~R. Molla, and G.~Pandurangan.
\newblock Distributed computation in dynamic networks via random walks.
\newblock {\em Theoretical Computer Science}, 581:45--66, 2015.

\bibitem{podc10}
A.~{Das Sarma}, D.~Nanongkai, G.~Pandurangan, and P.~Tetali.
\newblock Efficient distributed random walks with applications.
\newblock In {\em PODC}, pages 201--210, 2010.

\bibitem{jacm}
A.~{Das Sarma}, D.~Nanongkai, G.~Pandurangan, and P.~Tetali.
\newblock Distributed random walks.
\newblock {\em J. {ACM}}, 60(1):2:1--2:31, 2013.

\bibitem{dolev:stabilize}
S.~Dolev.
\newblock {\em Self-stabilization}.
\newblock MIT Press, 2000.

\bibitem{Dutta13}
C.~Dutta, G.~Pandurangan, R.~Rajaraman, Z.~Sun, and E.~Viola.
\newblock On the complexity of information spreading in dynamic networks.
\newblock In {\em ACM-SIAM SODA}, 2013.

\bibitem{gafni+b:link-reversal}
E.~Gafni and B.~Bertsekas.
\newblock Distributed algorithms for generating loop-free routes in networks
  with frequently changing topology.
\newblock {\em IEEE Transactions on Communications}, 1981.

\bibitem{haeupler:gossip}
B.~Haeupler.
\newblock Analyzing network coding gossip made easy.
\newblock In {\em ACM STOC}, pages 293--302, 2011.

\bibitem{haeupler+k:dynamic}
B.~Haeupler and D.~Karger.
\newblock Faster information dissemination in dynamic networks via network
  coding.
\newblock In {\em ACM PODC}, pages 381--390, 2011.

\bibitem{HK12}
B.~Haeupler and F.~Kuhn.
\newblock Lower bounds on information dissemination in dynamic networks.
\newblock In {\em DISC}, pages 166--180, 2012.

\bibitem{king}
V.~King, S.~Kutten, and M.~Thorup.
\newblock Construction and impromptu repair of an {MST} in a distributed
  network with o(m) communication.
\newblock In {\em PODC}, pages 71--80, 2015.

\bibitem{kuhn-stoc10}
F.~Kuhn, N.~Lynch, and R.~Oshman.
\newblock Distributed computation in dynamic networks.
\newblock In {\em STOC}, pages 513--522, 2010.

\bibitem{kuhn-survey}
F.~Kuhn and R.~Oshman.
\newblock Dynamic networks: Models and algorithms.
\newblock {\em SIGACT News}, 42(1), 2011.

\bibitem{kutten-jacm15}
S.~Kutten, G.~Pandurangan, D.~Peleg, P.~Robinson, and A.~Trehan.
\newblock On the complexity of universal leader election.
\newblock {\em Journal of ACM}, 62(1):7:1--7:27, 2015.

\bibitem{leighton:book}
F.~T. Leighton.
\newblock {\em Introduction to Parallel Algorithms and Architectures: Arrays,
  Trees, and Hypercubes}.
\newblock Morgan-Kaufmann, 1991.

\bibitem{lynch:distributed}
N.~A. Lynch.
\newblock {\em Distributed Algorithms}.
\newblock Morgan Kaufmann, 1996.

\bibitem{odell05}
R.~O'Dell and R.~Wattenhofer.
\newblock Information dissemination in highly dynamic graphs.
\newblock In {\em DIALM-POMC}, pages 104--110, 2005.

\bibitem{peleg}
D.~Peleg.
\newblock {\em Distributed computing: a locality-sensitive approach}.
\newblock SIAM, Philadelphia, PA, USA, 2000.

\bibitem{topkis85}
D.~M. Topkis.
\newblock Concurrent broadcast for information dissemination.
\newblock {\em IEEE Transactional Software Engineering}, 11(10):1107--1112,
  1985.

\end{thebibliography}

\end{document}